\documentclass[pre,10pt,onecolumn]{revtex4}%
\usepackage{amsfonts}
\usepackage{amsmath}
\usepackage{amssymb}
\usepackage{graphicx}%
\providecommand{\U}[1]{\protect\rule{.1in}{.1in}}
\newtheorem{theorem}{Theorem}

\newenvironment{proof}[1][Proof]{\noindent\textbf{#1.} }{\ \rule{0.5em}{0.5em}}
\begin{document}
\title{A combinatorial criterion for $k$-separability of multipartite Dicke states}
\author{Zhihua Chen$^{1}$}
\author{Zhihao Ma$^{2}$}
\author{Ting Gao$^{3}$}
\author{Simone Severini$^{4}$}
\affiliation{$^{1}$Department of Mathematics, College of Science, Zhejiang University of
Technology, Hangzhou 310024, China}
\affiliation{$^{2}$Department of Mathematics, Shanghai Jiaotong University, Shanghai
200240, China}
\affiliation{$^{3}$College of Mathematics and Information Science,Hebei Normal University,
Shijiazhuang 050024, China}
\affiliation{$^{4}$Department of Computer Science, University College London, Gower Street,
London WC1E 6BT, UK}

\begin{abstract}
We derive a combinatorial criterion for detecting $k$-separability of
$N$-partite Dicke states. The criterion is efficiently computable and
implementable without full state tomography.\textbf{ }We give examples in
which the criterion succeeds, where known criteria fail.

\end{abstract}
\maketitle

\section{Introduction}



While the structure of bipartite entangled states is fairly well-understood,
the study of multipartite entanglement still presents a number of partial
successes and difficult open problems (see the recent reviews
\cite{horodeckiqe,guehnewit}). General criteria for multipartite states of any
dimension were recently proposed in
\cite{HuberPRL,Guhne10,GaoPRA2010,GaoEPJD2011,HuberDicke,GaoPRA2012,GaoEPL2013,GaoPRL2014,Toth,
Huberqic}.

Here, we derive a criterion for detecting $k$-nonseparability in Dicke states
based on various ideas developed in
\cite{HuberPRL,Guhne10,GaoPRA2010,GaoEPJD2011,GaoPRA2012,GaoEPL2013,me}. The
criterion can be seen as a generalization of a method for detecting genuine
multipartite entanglement in Dicke states detailed in\textbf{ }%
\cite{HuberDicke}. The criterion have the advantages of being computationally
efficient and implementable without the need of state tomography. We give
examples in which the criterion is stronger than the ones proposed in
\cite{Huberqic}.

Let us recall some standard terminology and the definition of a Dicke state.
An $N$-partite pure state $|\psi\rangle\in\mathcal{H}_{1}\otimes
\mathcal{H}_{2}\otimes\cdots\otimes\mathcal{H}_{N}$ (dim $\mathcal{H}%
_{i}=d_{i}\geq2,1\leq i\leq N$) is said to be $k$\emph{-separable} if there is
a $k$-partition \cite{GaoPRA2012, GaoEPL2013}
\[%
\begin{tabular}
[c]{lll}%
$j_{1}^{1}\cdots j_{m_{1}}^{1}|j_{1}^{2}\cdots j_{m_{2}}^{2}|\cdots|j_{1}%
^{k}\cdots j_{m_{k}}^{k}$ & such that & $|\psi\rangle=|\psi_{1}\rangle
_{j_{1}^{1}\cdots j_{m_{1}}^{1}}|\psi_{2}\rangle_{j_{1}^{2}\cdots j_{m_{2}%
}^{2}}\cdots|\psi_{k}\rangle_{j_{1}^{k}\cdots j_{m_{k}}^{k}},$%
\end{tabular}
\
\]
where $|\psi_{i}\rangle_{j_{1}^{i}\cdots j_{m_{i}}^{i}}$ is the state of the
subsystems $j_{1}^{i},j_{2}^{i},...,j_{m_{i}}^{i}$, and $\bigcup_{i=1}%
^{k}\{j_{1}^{i},j_{2}^{i},\cdots,j_{m_{i}}^{i}\}=\{1,2,\cdots,N\}$. More
generally, an $N$-partite mixed state $\rho$ is said to be $k$%
\emph{-separable} if it can be written as a convex combination of
$k$-separable pure states $\rho=\sum\limits_{i}p_{i}|\psi_{i}\rangle
\langle\psi_{i}|$, where $|\psi_{i}\rangle$ is possibly $k$-separable under
different partitions. An $n$-partite state is said to be \emph{fully
separable} when it is $N$-separable and\emph{\ }$N$\emph{-partite entangled}
if it is not $2$-separable. A $k$-separable mixed state might not be separable
with regard to any specific $k$-partition, which makes $k$-separability
difficult to deal with. We shall consider pure states as a special case. 

The $N$\emph{-qubits Dicke state with }$m$\emph{\ excitations} (see
\cite{defdick}) is defined as
\[%
\begin{tabular}
[c]{lll}%
$|D_{m}^{N}\rangle=\dfrac{1}{\sqrt{C_{N}^{m}}}%
{\displaystyle\sum\limits_{1\leq i_{j}\neq i_{l}\leq N}}
|\phi_{i_{1},...,i_{m}}\rangle,$ & where & $|\phi_{i_{1},...,i_{m}}\rangle=%
{\displaystyle\bigotimes\limits_{i\not \in \{i_{1},...,i_{m}\}}}
|0\rangle_{i}%
{\displaystyle\bigotimes\limits_{i\in\{i_{1},...,i_{m}\}}}
|1\rangle_{i},$%
\end{tabular}
\
\]
where $C_{N}^{m}:=\binom{N}{m}$ is the binomial coefficient. For instance,
\[
|D_{2}^{4}\rangle=6^{-1/2}\left(  |1100\rangle+|1010\rangle+|1001\rangle
+|0110\rangle+|0011\rangle+|0101\rangle\right)  ,
\]
when $N=4$ and $m=2$.

Section II contains the statements and proofs of the results. Examples are in
Section III.

\section{Results}

We construct a set of inequalities which are optimally suited for testing
whether a given Dicke state is $N$-partite entangled:

\begin{theorem}
Suppose that $\rho$ is an $N$-partite density matrix acting on a Hilbert space
$\mathcal{H=H}_{1}\otimes\mathcal{H}_{2}\otimes\cdots\otimes\mathcal{H}_{N}$.
Let%
\[
\mathcal{F}(\rho,\phi):=A(\rho,\phi)-B(\rho,\phi),
\]
where%
\[
A(\rho,\phi):=\sum_{1\leq i\neq j\neq j^{\prime}\leq N}\left(  |\langle
\phi_{i,j}|\rho|\phi_{i,j^{\prime}}\rangle|-\sqrt{\langle\phi_{i,j}%
|\otimes\langle\phi_{i,j^{\prime}}|\Pi_{j}\rho^{\otimes2}\Pi_{j}|\phi
_{i,j}\rangle\otimes|\phi_{i,j^{\prime}}\rangle}\right)  ,
\]
and
\[
B(\rho,\phi):=N_{k}\sum_{1\leq i\neq j\leq N}\langle\phi_{i,j}|\rho|\phi
_{i,j}\rangle.
\]
Here, $|\phi_{i,j}\rangle:=|0...010...010...0\rangle\in\mathcal{H}$, with the
$1$s in the subspaces $\mathcal{H}_{i}$ and $\mathcal{H}_{j}$, $N_{k}%
:=\max\{2(N-k-1),N-k\}$ and $\Pi_{j}$, is the operator swapping the two copies
of $\mathcal{H}_{j}$ in $\mathcal{H}\otimes\mathcal{H}$, for $1\leq i\neq
j\leq N$. If the density matrix $\rho$ is $k$-separable then
\[
\mathcal{F}(\rho,\phi)\leq0.
\]

\end{theorem}

\begin{proof}
We start with a 4-qubit state to get an intuition. Note that for a four-qubit
pure state $\rho=|\psi\rangle\langle\psi|$, we have
\begin{align}
\mathcal{F}(\rho,\phi)  &  =2(|\rho_{4,6}|-\sqrt{\rho_{2,2}\rho_{8,8}}%
+|\rho_{4,7}|-\sqrt{\rho_{3,3}\rho_{8,8}}+|\rho_{4,10}|-\sqrt{\rho_{2,2}%
\rho_{12,12}}+|\rho_{4,11}|-\sqrt{\rho_{3,3}\rho_{12,12}}%
\nonumber\label{fpsir}\\
&  +|\rho_{6,7}|-\sqrt{\rho_{5,5}\rho_{8,8}}+|\rho_{6,10}|-\sqrt{\rho
_{2,2}\rho_{14,14}}+|\rho_{6,13}|-\sqrt{\rho_{5,5}\rho_{14,14}}+|\rho
_{7,11}|-\sqrt{\rho_{3,3}\rho_{15,15}}\nonumber\\
&  +|\rho_{7,13}|-\sqrt{\rho_{5,5}\rho_{15,15}}+|\rho_{10,11}|-\sqrt
{\rho_{9,9}\rho_{12,12}}+|\rho_{10,13}|-\sqrt{\rho_{9,9}\rho_{14,14}}\\
&  +|\rho_{11,13}|-\sqrt{\rho_{9,9}\rho_{15,15}})-N_{k}(\rho_{4,4}+\rho
_{6,6}+\rho_{7,7}+\rho_{10,10}+\rho_{11,11}+\rho_{13,13}).\nonumber
\end{align}
If a 4-qubit pure state $\rho=|\psi\rangle\langle\psi|$ is biseparable then
$\mathcal{F}(\rho,\phi)\leq0$ by the criterion in \cite{Huber11a}.

Suppose that a 4-qubit pure state $\rho=|\psi\rangle\langle\psi|$ with
$|\psi\rangle=\sum\limits_{i_{1}i_{2}i_{3}i_{4}}\psi_{i_{1}i_{2}i_{3}i_{4}%
}|i_{1}i_{2}i_{3}i_{4}\rangle$ $(i_{1},i_{2},i_{3},i_{4}=0,1)$ is
$k$-separable, where $k=3,4$. Then,
\begin{align*}
\mathcal{F}(|\psi\rangle,\phi)  &  =2(|\psi_{0011}\psi_{0101}|-|\psi
_{0001}\psi_{0111}|+|\psi_{0011}\psi_{0110}|-|\psi_{0010}\psi_{0111}%
|+|\psi_{0011}\psi_{1001}|-|\psi_{0001}\psi_{1011}|\\
&  +|\psi_{0011}\psi_{1010}|-|\psi_{0010}\psi_{1011}|+|\psi_{0101}\psi
_{0110}|-|\psi_{0100}\psi_{0111}|+|\psi_{0101}\psi_{1001}|-|\psi_{0001}%
\psi_{1101}|\\
&  +|\psi_{0101}\psi_{1100}|-|\psi_{0100}\psi_{1101}|+|\psi_{0110}\psi
_{1010}|-|\psi_{0010}\psi_{1110}|+|\psi_{0110}\psi_{1100}|-|\psi_{0100}%
\psi_{1110}|\\
&  +\psi_{1001}\psi_{1010}|-|\psi_{1011}\psi_{1000}|+|\psi_{1001}\psi
_{1100}|-|\psi_{1000}\psi_{1101}|+|\psi_{1010}\psi_{1100}|-|\psi_{1000}%
\psi_{1110}|)\\
&  -N_{k}(|\psi_{0011}|^{2}+|\psi_{0101}|^{2}+|\psi_{0110}|^{2}+|\psi
_{1001}|^{2}+|\psi_{1010}|^{2}+|\psi_{1100}|^{2}).
\end{align*}
Note that there are six 3-partitions $1|2|34$, $1|3|24$, $1|4|23$, $2|3|14$,
$2|4|13$, and $3|4|12$. WLOG we prove that $\mathcal{F}(|\psi\rangle,\phi
)\leq0$ holds for a 4-qubit pure state $\rho=|\psi\rangle\langle\psi|$ which
is 3-separable under the partition $1|2|34$. Suppose that
\begin{align*}
|\psi\rangle &  =(a_{1}|0\rangle+a_{2}|1\rangle)_{1}\otimes(b_{1}%
|0\rangle+b_{2}|1\rangle)_{2}\otimes(c_{1}|00\rangle+c_{2}|01\rangle
+c_{3}|10\rangle+c_{4}|11\rangle)_{34}\\
&  =a_{1}b_{1}c_{1}|0000\rangle+a_{1}b_{1}c_{2}|0001\rangle+a_{1}b_{1}%
c_{3}|0010\rangle+a_{1}b_{1}c_{4}|0011\rangle+a_{1}b_{2}c_{1}|0100\rangle
+a_{1}b_{2}c_{2}|0101\rangle\\
&  +a_{1}b_{2}c_{3}|0110\rangle+a_{1}b_{2}c_{4}|0111\rangle+a_{2}b_{1}%
c_{1}|1000\rangle+a_{2}b_{1}c_{2}|1001\rangle+a_{2}b_{1}c_{3}|1010\rangle
+a_{2}b_{1}c_{4}|1011\rangle\\
&  +a_{2}b_{2}c_{1}|1100\rangle+a_{2}b_{2}c_{2}|1101\rangle+a_{2}b_{2}%
c_{3}|1110\rangle+a_{2}b_{2}c_{4}|1111\rangle,
\end{align*}
then
\begin{align*}
A_{1}  &  =2(|\psi_{0011}\psi_{0101}|-|\psi_{0001}\psi_{0111}|+|\psi
_{0011}\psi_{0110}|-|\psi_{0010}\psi_{0111}|+|\psi_{0011}\psi_{1001}%
|-|\psi_{0001}\psi_{1011}|\\
&  +|\psi_{0011}\psi_{1010}|-|\psi_{0010}\psi_{1011}|+|\psi_{0101}\psi
_{1001}|-|\psi_{0001}\psi_{1101}|+|\psi_{0101}\psi_{1100}|-|\psi_{0100}%
\psi_{1101}|+|\psi_{0110}\psi_{1010}|\\
&  -|\psi_{0010}\psi_{1110}|+|\psi_{0110}\psi_{1100}|-|\psi_{0100}\psi
_{1110}|+|\psi_{1001}\psi_{1100}|-|\psi_{1000}\psi_{1101}|+|\psi_{1010}%
\psi_{1100}|-|\psi_{1000}\psi_{1110}|)\\
&  =0;
\end{align*}
and
\begin{align*}
A_{2}  &  =2(|\psi_{0101}\psi_{0110}|-|\psi_{0100}\psi_{0111}|+\psi_{1001}%
\psi_{1010}|-|\psi_{1011}\psi_{1000}|)\\
&  -(|\psi_{0011}|^{2}+|\psi_{0101}|^{2}+|\psi_{0110}|^{2}+|\psi_{1001}%
|^{2}+|\psi_{1010}|^{2}+|\psi_{1100}|^{2})\\
&  \leq0.
\end{align*}
It follows that $\mathcal{F}(|\psi\rangle,\phi)=A_{1}+A_{2}\leq0$. if
$|\psi\rangle$ is 3-separable.

If $|\psi\rangle$ is fully separable, then
\begin{align*}
|\psi\rangle &  =(a_{1}|0\rangle+a_{2}|1\rangle)_{1}\otimes(b_{1}%
|0\rangle+b_{2}|1\rangle)_{2}\otimes(c_{1}|0\rangle+c_{2}|1\rangle)_{3}%
\otimes(d_{1}|0\rangle+d_{2}|1\rangle)_{4}\\
&  =a_{1}b_{1}c_{1}d_{1}|0000\rangle+a_{1}b_{1}c_{1}d_{2}|0001\rangle
+a_{1}b_{1}c_{2}d_{1}|0010\rangle+a_{1}b_{1}c_{2}d_{2}|0011\rangle+a_{1}%
b_{2}c_{1}d_{1}|0100\rangle+a_{1}b_{2}c_{1}d_{2}|0101\rangle\\
&  +a_{1}b_{2}c_{2}d_{1}|0110\rangle+a_{2}b_{2}c_{2}d_{2}|0111\rangle
+a_{2}b_{1}c_{1}d_{1}|1000\rangle+a_{2}b_{2}c_{1}d_{2}|1001\rangle+a_{2}%
b_{1}c_{2}d_{1}|1010\rangle+a_{2}b_{1}c_{2}d_{2}|1011\rangle\\
&  +a_{2}b_{2}c_{1}d_{1}|1100\rangle+a_{2}b_{2}c_{1}d_{2}|1101\rangle
+a_{2}b_{2}c_{2}d_{1}|1110\rangle+a_{2}b_{2}c_{2}d_{2}|1111\rangle,
\end{align*}
and
\begin{align*}
\mathcal{F}(|\psi\rangle,\phi)  &  =2(|\psi_{0011}\psi_{0101}|-|\psi
_{0001}\psi_{0111}|+|\psi_{0011}\psi_{0110}|-|\psi_{0010}\psi_{0111}%
|+|\psi_{0011}\psi_{1001}|-|\psi_{0001}\psi_{1011}|\\
&  +|\psi_{0011}\psi_{1010}|-|\psi_{0010}\psi_{1011}|+|\psi_{0101}\psi
_{0110}|-|\psi_{0100}\psi_{0111}|+|\psi_{0101}\psi_{1001}|-|\psi_{0001}%
\psi_{1101}|\\
&  +|\psi_{0101}\psi_{1100}|-|\psi_{0100}\psi_{1101}|+|\psi_{0110}\psi
_{1010}|-|\psi_{0010}\psi_{1110}|+|\psi_{0110}\psi_{1100}|-|\psi_{0100}%
\psi_{1110}|\\
&  +|\psi_{1001}\psi_{1010}|-|\psi_{1000}\psi_{1011}|+|\psi_{1001}\psi
_{1100}|-|\psi_{1000}\psi_{1101}|+|\psi_{1010}\psi_{1100}|-|\psi_{1000}%
\psi_{1110}|)\\
&  =0.
\end{align*}
The equalities above confirms the statement in Eq.(\ref{fpsir}), when
restricted to 4-qubit pure states.

For the general case, we use the notation and proof method given in
\cite{GaoPRA2010, GaoEPL2013}.

Suppose that $\rho=|\psi\rangle\langle\psi|$ is a $k$-separable pure state
under the partition of $\{1,2,\cdots,N\}$ into $k$ pairwise disjoint subsets:
$\{1,2,\cdots,N\}=\bigcup_{l=1}^{k}A_{l}$, with $A_{l}=\{j_{1}^{l},j_{2}%
^{l},\cdots,j_{m_{l}}^{l}\}$ and%
\begin{align*}
|\psi\rangle &  =|\psi_{1}\rangle_{j_{1}^{1}\cdots j_{m_{1}}^{1}}\cdots
|\psi_{k}\rangle_{j_{1}^{k}\cdots j_{m_{k}}^{k}}\\
&  =\left(  \sum\limits_{i_{1}^{1},\cdots,i_{m_{1}}^{1}}a_{i_{1}^{1}\cdots
i_{m_{1}}^{1}}|i_{1}^{1}\cdots i_{m_{1}}^{1}\rangle\right)  _{j_{1}^{1}\cdots
j_{m_{1}}^{1}}\cdots\left(  \sum\limits_{i_{1}^{k},\cdots,i_{m_{k}}^{k}%
}a_{i_{1}^{k}\cdots i_{m_{k}}^{k}}|i_{1}^{k}\cdots i_{m_{k}}^{k}%
\rangle\right)  _{j_{1}^{k}\cdots j_{m_{k}}^{k}}\\
&  \sum\limits_{i_{1}^{1},\cdots,i_{m_{1}}^{1},\cdots,i_{1}^{k},\cdots
,i_{m_{k}}^{k}}a_{i_{1}^{1}\cdots i_{m_{1}}^{1}}\cdots a_{i_{1}^{k}\cdots
i_{m_{k}}^{k}}|i_{1}^{1}\cdots i_{m_{1}}^{1}\cdots i_{1}^{k}\cdots i_{m_{k}%
}^{k}\rangle_{j_{1}^{1}\cdots j_{m_{1}}^{1}\cdots j_{1}^{k}\cdots j_{m_{k}%
}^{k}}.
\end{align*}
Hence,
\[
\rho_{\sum\limits_{s,t}i_{t}^{s}d_{j_{t}^{s}+1}d_{j_{t}^{s}+2}\cdots
d_{N}d_{N+1}+1,\sum\limits_{s,t}\widetilde{i_{t}^{s}}d_{j_{t}^{s}+1}%
d_{j_{t}^{s}+2}\cdots d_{N}d_{N+1}+1}=a_{i_{1}^{1}\cdots i_{m_{1}}^{1}}\cdots
a_{i_{1}^{k}\cdots i_{m_{k}}^{k}}a_{\widetilde{i_{1}^{1}}\cdots
\widetilde{i_{m_{1}}^{1}}}^{\ast}\cdots a_{\widetilde{i_{1}^{k}}%
\cdots\widetilde{i_{m_{k}}^{k}}}^{\ast}.
\]
The sum is over all possible values of $\{i_{t}^{s}|s\in\{1,2,\cdots
,k\},t\in\{1,2,\cdots,m_{s}\}\}$, $d_{i}=2$, when $i\neq N+1$ and $d_{N+1}=1$.

We shall distinguish between the cases in which both indices $j$ and
$j^{\prime}$ correspond to different parts $A_{l}$ and $A_{l}^{\prime}$, or
the same parts $A_{l}$, $1\leq l\neq l^{\prime}\leq k$, with respect to
$|\psi\rangle$. By direct calculation, one has the following:
\begin{equation}
|\langle\phi_{i,j}|\rho|\phi_{i,j^{\prime}}\rangle|=\sqrt{\langle\phi
_{i,j}|\rho|\phi_{i,j}\rangle\langle\phi_{i,j^{\prime}}|\rho|\phi
_{i,j^{\prime}}\rangle}\leq\frac{\langle\phi_{i,j}|\rho|\phi_{i,j}%
\rangle+\langle\phi_{i,j^{\prime}}|\rho|\phi_{i,j^{\prime}}\rangle}{2},
\label{iandj}%
\end{equation}
when $j$ and $j^{\prime}$ are in the same part;
\begin{equation}
|\langle\phi_{i,j}|\rho|\phi_{i,j^{\prime}}\rangle|=\sqrt{\langle\phi_{i}%
|\rho|\phi_{i}\rangle\langle\phi_{i,j,j^{\prime}}|\rho|\phi_{i,j,j^{\prime}%
}\rangle}=\sqrt{\langle\phi_{i,j}|\otimes\langle\phi_{i,j}|\Pi_{j^{\prime}%
}^{+}\rho^{\otimes2}\Pi_{j}|\phi_{i,j}\rangle\otimes|\phi_{i,j^{\prime}%
}\rangle}, \label{iorj}%
\end{equation}
when $j$ and $j^{\prime}$ are in the different parts ($j\in A_{l},j\in
A_{l^{\prime}}$ with $l\neq l^{\prime}$). Here, $|\phi_{i}\rangle
=|00\cdots010\cdots0\rangle$, with $|1\rangle$ in the $i$-th subspace $H_{i}$,
and $|\phi_{i,j,j^{\prime}}\rangle=|0\cdots010\cdots010\cdots010\cdots
0\rangle$, such that all subspaces are in the state $|0\rangle$, except for
the subspaces $H_{i}$, $H_{j}$ and $H_{j}^{\prime}$, which are in the state
$|1\rangle$.

For a given $|\phi_{i,j}\rangle$, the number of $|\phi_{i,j^{\prime}}\rangle
$'s, with $j$ and $j^{\prime}$ in same part, is at most $\max\{2(N-k-1),N-k\}$%
. Notice that the maximal number of subsystems contained in a part of a
$k$-partition is $N-k+1$. Suppose that $A_{1}|A_{2}|\cdots|A_{k}$ is a
$k$-partition of $\{1,2,\cdots,n\}$, where $A_{l}=\{j_{1}^{l}\}$, for
$l=1,2,\cdots,k-1$, and $A_{k}=\{j_{1}^{k},j_{2}^{k},\cdots,j_{N-k+1}^{k}\}$.
When $i$, $j$ and $j^{\prime}$ are in the same part $A_{k}$, the number of
$|\phi_{i,j^{\prime}}\rangle$'s is $2(N-k-1)$. When $i$ belongs to $A_{1}\cup
A_{2}\cup\cdots\cup A_{k-1}$, while $j$ and $j^{\prime}$ belong to $A_{k}$,
the number of $|\phi_{i,j^{\prime}}\rangle$'s is $N-k$. Therefore, the number
of $|\phi_{i,j^{\prime}}\rangle$'s satisfying $j$ and $j^{\prime}$ in the same
part is at most $\max\{2(N-k-1),N-k\}$. This number is denoted as $N_{k}$.

By using the inequalities in (\ref{iandj}) and (\ref{iorj}), we have
\begin{align*}
\sum\limits_{1\leq i,j,j^{\prime}\leq N}|\langle\phi_{i,j}|\rho|\phi
_{i,j^{\prime}}\rangle| &  =\sum\limits_{i}\sum\limits_{\substack{j\in
A_{l},j^{\prime}\in A_{l^{\prime}},l\neq{l^{\prime}}\\l,l^{\prime}%
\in\{1,2,\cdots,k\}}}|\langle\phi_{i,j}|\rho|\phi_{i,j}\rangle|+\sum
\limits_{\substack{j,j^{\prime}\in A_{l},j\neq j^{\prime}\\l\in\{1,2,\cdots
,k\}}}|\langle\phi_{i,j}|\rho|\phi_{i,j}\rangle|\\
&  \leq\sum\limits_{i}\sum\limits_{\substack{j\in A_{l},j\in A_{l^{\prime}%
},l\neq{l^{\prime}}\\l,l^{\prime}\in\{1,2,\cdots,k\}}}\sqrt{\langle\phi
_{i,j}|\otimes\langle\phi_{i,j^{\prime}}|\Pi_{j}^{+}\rho^{\otimes2}\Pi
_{j}|\phi_{i,j}\rangle\otimes|\phi_{i,j^{\prime}}\rangle}\\
&  +\sum\limits_{i}\sum\limits_{\substack{j,j^{\prime}\in A_{l},j\neq
j^{\prime}\\l\in\{1,2,\cdots,k\}}}\left(  \frac{\langle\phi_{i,j}|\rho
|\phi_{i,j}\rangle+\langle\phi_{i,j^{\prime}}|\rho|\phi_{i,j^{\prime}}\rangle
}{2}\right)  \\
&  \leq\sum\limits_{i}\sum\limits_{j\neq j^{\prime}}\sqrt{\langle\phi
_{i,j}|\otimes\langle\phi_{i,j^{\prime}}|\Pi_{j}^{+}\rho^{\otimes2}\Pi
_{j}|\phi_{i,j}\rangle\otimes|\phi_{i,j^{\prime}}\rangle}+N_{k}\sum
\limits_{i,j}\langle\phi_{i,j}|\rho|\phi_{i,j}\rangle.
\end{align*}
Thus, the inequality in the statement of the theorem is satisfied by all
$k$-separable $N$-partite pure states.

It remains to show that the inequality holds if $\rho$ is a $k$-separable
$N$-partite mixed state. Indeed, the generalization of the inequality to mixed
states is a direct consequence of the convexity of the first summation in
$A(\rho,\phi)$, the concavity of $B(\rho,\phi)$, and the second summation in
$A(\rho,\phi)$, which we can see as follows.

Suppose that
\[
\rho=\sum\limits_{m}p_{m}\rho_{m}=\sum\limits_{m}p_{m}|\psi_{m}\rangle
\langle\psi_{m}|
\]
is a $k$-separable $N$-partite mixed state, where $\rho_{m}=|\psi_{m}%
\rangle\langle\psi_{m}|$ is $k$-separable. Then, by the Cauchy-Schwarz
inequality, $(\sum_{k=1}^{m}x_{k}y_{k})^{2}\leq(\sum_{k=1}^{m}x_{k}^{2}%
)(\sum_{k=1}^{m}y_{k}^{2})$, we get
\begin{align*}
\sum\limits_{i}\sum\limits_{j\neq j^{\prime}}|\langle\phi_{i,j}|\rho
|\phi_{i,j^{\prime}}\rangle|  &  \leq\sum\limits_{i}\sum\limits_{m}p_{m}%
\sum\limits_{j\neq j^{\prime}}|\langle\phi_{i,j}|\rho_{m}|\phi_{i,j}\rangle|\\
&  \leq\sum\limits_{m}p_{m}\left(  \sum\limits_{i}\sum\limits_{j\neq
j^{\prime}}\sqrt{\langle\phi_{i,j}|\otimes\langle\phi_{i,j}|\Pi_{j}^{+}%
\rho_{m}^{\otimes2}\Pi_{j}|\phi_{i,j^{\prime}}\rangle\otimes|\phi
_{i,j^{\prime}}\rangle}\right. \\
&  \left.  +N_{k}\sum\limits_{i,j}\sqrt{\langle\phi_{i,j}|\otimes\langle
\phi_{i,j}|\Pi_{j}^{+}\rho_{m}^{\otimes2}\Pi_{j}|\phi_{i,j}\rangle\otimes
|\phi_{i,j}\rangle}\right) \\
&  =\sum\limits_{i}\sum\limits_{j\neq j^{\prime}}\sum\limits_{m}\sqrt
{\langle\phi_{i}|p_{m}\rho_{m}|\phi_{i}\rangle}\sqrt{\langle\phi
_{i,j,j^{\prime}}|p_{m}\rho_{m}|\phi_{i,j,j^{\prime}}\rangle}+N_{k}%
\sum\limits_{i,j}\sum\limits_{m}p_{m}\langle\phi_{i,j}|\rho_{m}|\phi
_{i,j}\rangle\\
&  \leq\sum\limits_{i}\sum\limits_{j\neq j^{\prime}}\sqrt{\sum\limits_{m}%
\langle\phi_{i}|p_{m}\rho_{m}|\phi_{i}\rangle\sum\limits_{m}\langle
\phi_{i,j,j^{\prime}}|p_{m}\rho_{m}|\phi_{i,j,j^{\prime}}\rangle}+N_{k}%
\sum\limits_{i,j}\langle\phi_{i,j}|\rho|\phi_{i,j}\rangle\\
&  =\sum\limits_{i}\sum\limits_{j\neq j^{\prime}}\sqrt{\langle\phi
_{i,j}|\otimes\langle\phi_{i,j}|\Pi_{j}^{+}\rho^{\otimes2}\Pi_{j}%
|\phi_{i,j^{\prime}}\rangle\otimes|\phi_{i,j^{\prime}}\rangle}+N_{k}%
\sum\limits_{i,j}\sqrt{\langle\phi_{i,j}|\otimes\langle\phi_{i,j}|\Pi_{j}%
^{+}\rho^{\otimes2}\Pi_{j}|\phi_{i,j}\rangle\otimes|\phi_{i,j}\rangle},
\end{align*}
as desired. This completes the proof.
\end{proof}

\bigskip


We can choose $|\phi\rangle$ \emph{ad hoc }to get different inequalities for
detecting $k$-separability of different classes. For Theorem 2, we have chosen
$|\phi\rangle$ to be an $N$-qubit product states with $m$ excitations
(\emph{i.e.} $m$ entries of $|\phi\rangle$ are $|1\rangle$, while the
remaining $N-m$ entries are $|0\rangle$). The criterion performs well to
detect $k$-separability for $N$ qubit Dicke states with $m$ excitations mixed
with white noises.

\begin{theorem}
Suppose that $\rho$ is an $N$-partite density matrix acting on Hilbert space
$\mathcal{H=H}_{1}\otimes\mathcal{H}_{2}\otimes\cdots\otimes\mathcal{H}_{N}$,
and $|\phi_{i_{1}i_{2},\cdots,i_{m}}\rangle=|0\cdots010\cdots010\cdots
010\cdots0\rangle$ is a state of $\mathcal{H}$, where the local state in
$\mathcal{H}_{l}$ is $|0\rangle$, for $l\neq i_{1},i_{2},\cdots,i_{m}$, and
$|1\rangle$, for $l=i_{1},i_{2},\cdots,i_{m}$. Let
\[
\mathcal{F}(\rho,\phi):=A(\rho,\phi)-B(\rho,\phi),
\]
with
\begin{align*}
A(\rho,\phi) &  :=\sum\limits_{i_{1},\cdots,i_{j},\cdots,i_{m},i_{j}%
^{^{\prime}}}\left(  |\langle\phi_{i_{1}\cdots,i_{j},\cdots,i_{m}}|\rho
|\phi_{i_{1},\cdots,i_{j}^{^{\prime}}\cdots,i_{m}}\rangle|\right.  \\
&  \left.  -\sqrt{\langle\phi_{i_{1},\cdots,i_{j},\cdots,i_{m}}|\otimes
\langle\phi_{i_{1},\cdots,i_{j}^{^{\prime}},\cdots,i_{m}}|\Pi_{i_{j}}%
\rho^{\otimes2}\Pi_{i_{j}}|\phi_{i_{1},\cdots,i_{j},\cdots,i_{m}}%
\rangle\otimes|\phi_{i_{1},\cdots,i_{j}^{^{\prime}},\cdots,i_{m}}\rangle
}\right)  ,
\end{align*}
and
\[
B(\rho,\phi):=N_{k}\sum\limits_{i_{1},i_{2},\cdots,i_{m}}\langle\phi
_{i_{1},\cdots,i_{j},\cdots,i_{m}}|\rho|\phi_{i_{1},\cdots,i_{j},\cdots,i_{m}%
}\rangle.
\]
Here, $\Pi_{i_{j}}$ is the operator swapping the two copies of $\mathcal{H}%
_{i_{j}}$ in the twofold copy Hilbert space $\mathcal{H}^{\otimes
2}:=\mathcal{H}\otimes\mathcal{H}$, and
\[
N_{k}:=\max\{m(N-k+1-m),(m-1)(N-k-m+2),\cdots,(N-k)\}.
\]
If the density matrix $\rho$ is $k$-separable then
\[
\mathcal{F}(\rho,\phi)\leq0.
\]

\end{theorem}


In the following statement, we consider a criterion which is suitable for any
general quantum states. The states to be chosen are $|\chi\rangle$,
$|\chi_{\alpha}\rangle$ and $|\chi_{\beta}\rangle$.

\begin{theorem}
Let $V=\{|\chi_{1}\rangle,...,|\chi_{m}\rangle\}$ be a set of product states
in $\mathcal{H=H}_{1}\otimes\mathcal{H}_{2}\otimes\cdots\otimes\mathcal{H}%
_{N}$. If $\rho$ is $k$-separable then%
\begin{align}
\mathcal{T}(\rho,\chi) &  =\sum\limits_{|\chi_{\alpha}\rangle\in V}%
\sum\limits_{|\chi_{\beta}\rangle\in K_{\alpha}}\left(  |\langle\chi_{\alpha
}|\rho|\chi_{\beta}\rangle|-\sqrt{\langle\chi_{\alpha}|\otimes\langle
\chi_{\beta}|\Pi_{\alpha\beta}\rho^{\otimes2}\Pi_{\alpha\beta}|\chi_{\alpha
}\rangle\otimes|\chi_{\beta}\rangle}\right)  -N_{k}\sum\limits_{\alpha}%
\langle\chi_{\alpha}|\rho|\chi_{\alpha}\rangle\nonumber\\
&  \leq0,\label{eqth2}%
\end{align}
where%
\[
K_{\alpha}:=\{|\chi_{\beta}\rangle:||\chi_{\alpha}\rangle\cap|\chi_{\beta
}\rangle|=N-2\text{ with }|\chi_{\alpha}\rangle,|\chi_{\beta}\rangle\in V\},
\]
and $||\chi_{\alpha}\rangle\cap|\chi_{\beta}\rangle|$ is the number of
coordinates that are equal in both vectors (\emph{i.e}., $|\chi_{\alpha
}\rangle$ and $|\chi_{\beta}\rangle$ have only two different local states, say
the $i_{\alpha\beta}$-th and $i_{\alpha\beta}^{\prime}$-th local states),
while $\Pi_{\alpha\beta}$ is the operator swapping the two copies of
$\mathcal{H}_{i_{\alpha\beta}}$ in $\mathcal{H}^{\otimes2}$. Additionally,
\[
N_{k}:=\max\limits_{\alpha,i_{1},i_{2},\cdots,i_{N-k+1}}s_{\alpha,i_{1}%
,i_{2},\cdots,i_{N-k+1}},
\]
where $s_{\alpha,i_{1},i_{2},\cdots,i_{N-k+1}}$ is the number of states
$|\chi_{\beta}\rangle$ in $K_{\alpha}$ such that two of the states for the
$N-k+1$ particles $i_{1},i_{2},\cdots,i_{N-k+1}$ in $|\chi_{\beta}\rangle$ are
different from that of $|\chi_{\alpha}\rangle$, when $K_{\alpha}\neq\emptyset$.
\end{theorem}

\begin{proof}
By using the same proof method as in \cite{GaoPRA2010, GaoEPL2013}, we prove
that (\ref{eqth2}) holds for any $k$-separable pure states $\rho=|\psi
\rangle\langle\psi|$. Let $\mathcal{T}(\rho,\chi)=A_{1}+A_{2}$, where $A_{1}$
is the sum of terms $|\langle\chi_{\alpha}|\rho|\chi_{\beta}\rangle
|-\sqrt{\langle\chi_{\alpha}|\otimes\langle\chi_{\beta}|\Pi_{\alpha\beta}%
\rho^{\otimes2}\Pi_{\alpha\beta}|\chi_{\alpha}\rangle\otimes|\chi_{\beta
}\rangle}$ in the first summation in (\ref{eqth2}). In this expression, the
two different bits of $|\chi_{\alpha}\rangle$ and $|\chi_{\beta}\rangle$ are
in two different parts of a $k$-partition, while $A_{2}$ is the sum containing
the summands in (\ref{eqth2}), such that the different bits of $|\chi_{\alpha
}\rangle$ and $|\chi_{\beta}\rangle$ are in the same part of a $k$-partition.

We first prove that $\mathcal{T}(\rho,\chi)\leq0$ for any $4$-partite pure
state. Let $V=\{|\chi_{1}\rangle,...,|\chi_{4}\rangle\}$ be a set of product
states in $\mathcal{H}$, where $|\chi_{1}\rangle=|0011\rangle$, $|\chi
_{2}\rangle=|0101\rangle$, $|\chi_{3}\rangle=|0110\rangle$, and $|\chi
_{4}\rangle=|1010\rangle$. Then $K_{1}=\{|\chi_{2}\rangle,|\chi_{3}%
\rangle,|\chi_{4}\rangle\}$, $K_{2}=\{|\chi_{1}\rangle,|\chi_{3}\rangle\}$,
$K_{3}=\{|\chi_{1}\rangle,|\chi_{2}\rangle,|\chi_{4}\rangle\}$, and
$K_{4}=\{|\chi_{1}\rangle,|\chi_{3}\rangle\}$. Thus,%
\begin{align*}
\mathcal{T}(\rho,\chi) &  =2\left(  \sum\limits_{i=2}^{4}\left(  |\langle
\chi_{1}|\rho|\chi_{i}\rangle|-\sqrt{\langle\chi_{1}|\otimes\langle\chi
_{i}|\Pi_{12}\rho^{\otimes2}\Pi_{12}|\chi_{1}\rangle\otimes|\chi_{i}\rangle
}\right)  \right.  \\
&  +|\langle\chi_{2}|\rho|\chi_{3}\rangle|-\sqrt{\langle\chi_{2}%
|\otimes\langle\chi_{3}|\Pi_{23}\rho^{\otimes2}\Pi_{23}|\chi_{2}\rangle
\otimes|\chi_{3}\rangle}\\
&  +\left.  |\langle\chi_{3}|\rho|\chi_{4}\rangle|-\sqrt{\langle\chi
_{3}|\otimes\langle\chi_{4}|\Pi_{34}\rho^{\otimes2}\Pi_{34}|\chi_{3}%
\rangle\otimes|\chi_{4}\rangle}\right)  -N_{k}\sum\limits_{i=1}^{4}\langle
\chi_{i}|\rho|\chi_{i}\rangle\\
&  =2(|\phi_{0011}\phi_{0101}|-|\phi_{0001}\phi_{0111}|+|\phi_{0011}%
\phi_{0110}|-|\phi_{0010}\phi_{0111}|+|\phi_{0011}\phi_{1010}|-|\phi
_{0010}\phi_{1011}|\\
&  +|\phi_{0110}\phi_{0101}|-|\phi_{0100}\phi_{0111}|+|\phi_{0110}\phi
_{1010}|-|\phi_{0010}\phi_{1110}|)\\
&  -N_{k}(|\phi_{0011}|^{2}+|\phi_{0101}|^{2}+|\phi_{0110}|^{2}+|\phi
_{1010}|^{2})\\
&  =A_{1}+A_{2};
\end{align*}

When $k=3$, there are six $3$-partitions, \emph{i.e.}, $1|2|34$, $1|3|24$,
$1|4|23$, $2|3|14$, $2|4|13$, and $3|4|12$. For $\chi_{1}$, we have
$s_{1,34}=0$, $s_{1,24}=1$, $s_{1,23}=1$, $s_{1,14}=1$, $s_{1,13}=0$, and
$s_{1,12}=0$; for $\chi_{2}$, we have $s_{2,34}=1$, $s_{2,24}=0$, $s_{2,23}%
=1$, $s_{2,14}=0$, $s_{2,13}=0$, and $s_{2,12}=0$; for $\chi_{3}$, we have
$s_{3,34}=1$, $s_{3,24}=1$, $s_{3,23}=0$, $s_{3,14}=0$, $s_{3,13}=0$, and
$s_{3,12}=1$; for $\chi_{4}$, we have $s_{4,34}=0$, $s_{4,24}=0$, $s_{4,23}%
=0$, $s_{4,14}=1$, $s_{4,13}=0$, and $s_{4,12}=1$. So, we get $N_{3}=1$.

For the case $1|2|34$:
\begin{align*}
A_{1}  &  =2\left[  \sum\limits_{i=2}^{4}\left(  |\langle\chi_{1}|\rho
|\chi_{i}\rangle|-\sqrt{\langle\chi_{1}|\otimes\langle\chi_{i}|\Pi_{12}%
\rho^{\otimes2}\Pi_{12}|\chi_{1}\rangle\otimes|\chi_{i}\rangle}\right)
+|\langle\chi_{3}|\rho|\chi_{4}\rangle|-\sqrt{\langle\chi_{3}|\otimes
\langle\chi_{4}|\Pi_{34}\rho^{\otimes2}\Pi_{34}|\chi_{3}\rangle\otimes
|\chi_{4}\rangle}\right] \\
&  =2(|\phi_{0011}\phi_{0101}|-|\phi_{0001}\phi_{0111}|+|\phi_{0011}%
\phi_{0110}|-|\phi_{0010}\phi_{0111}|+|\phi_{0011}\phi_{1010}|-|\phi
_{0010}\phi_{1011}|+|\phi_{0110}\phi_{1010}|-|\phi_{0010}\phi_{1110}|)\\
&  =0;
\end{align*}%
\begin{align*}
A_{2}  &  =2\left(  |\langle\chi_{2}|\rho|\chi_{3}\rangle|-\sqrt{\langle
\chi_{2}|\otimes\langle\chi_{3}|\Pi_{23}\rho^{\otimes2}\Pi_{23}|\chi
_{2}\rangle\otimes|\chi_{3}\rangle}\right)  -N_{3}(|\phi_{0011}|^{2}%
+|\phi_{0101}|^{2}+|\phi_{0110}|^{2}+|\phi_{1010}|^{2})\\
&  =2(|\phi_{0110}\phi_{0101}|-|\phi_{0100}\phi_{0111}|)-(|\phi_{0011}%
|^{2}+|\phi_{0101}|^{2}+|\phi_{0110}|^{2}+|\phi_{1010}|^{2})\\
&  \leq0.
\end{align*}
This implies that $\mathcal{T}(\rho,\chi)=A_{1}+A_{2}\leq0$. For the other
3-partitions, we can get the same result $\mathcal{T}(\rho,\chi)=A_{1}%
+A_{2}\leq0$, as $1|2|34$.

When $k=4$, there is a single $4$-partition, $1|2|3|4$. Then, it is not
possible for any two different bits to be in the same partition. It follows
that $N_{4}=0$ and $\mathcal{T}(\rho,\chi)=A_{1}=0$.

For a $k$-separable $4$-partite mixed state $\rho=\sum\limits_{m}p_{m}\rho
_{m}$, where $\rho_{m}=|\psi_{m}\rangle\langle\psi_{m}|$ is $k$-separable, we
have
\begin{align*}
\mathcal{T}(\rho,\chi)  &  =2\left[  \sum\limits_{i=2}^{4}\left(  |\langle
\chi_{1}|\sum\limits_{m}p_{m}\rho_{m}|\chi_{i}\rangle|-\sqrt{\langle\chi
_{1}|\otimes\langle\chi_{i}|\Pi_{12}(\sum\limits_{m}p_{m}\rho_{m})^{\otimes
2}\Pi_{12}|\chi_{1}\rangle\otimes|\chi_{i}\rangle}\right)  \right. \\
&  +|\langle\chi_{2}|\sum\limits_{m}p_{m}\rho_{m}|\chi_{3}\rangle
|-\sqrt{\langle\chi_{2}|\otimes\langle\chi_{3}|\Pi_{23}(\sum\limits_{m}%
p_{m}\rho_{m})^{\otimes2}\Pi_{23}|\chi_{2}\rangle\otimes|\chi_{3}\rangle}\\
&  +\left.  |\langle\chi_{3}|\sum\limits_{m}p_{m}\rho_{m}|\chi_{4}%
\rangle|-\sqrt{\langle\chi_{3}|\otimes\langle\chi_{4}|\Pi_{34}(\sum
\limits_{m}p_{m}\rho_{m})^{\otimes2}\Pi_{34}|\chi_{3}\rangle\otimes|\chi
_{4}\rangle}\right] \\
&  -N_{k}\sum\limits_{i=1}^{4}\langle\chi_{i}|\sum\limits_{m}p_{m}\rho
_{m}|\chi_{i}\rangle\\
&  \leq\sum\limits_{m}2p_{m}\left[  \sum\limits_{i=2}^{4}\left(  |\langle
\chi_{1}|\rho_{m}|\chi_{i}\rangle|-\sqrt{\langle\chi_{1}|\otimes\langle
\chi_{i}|\Pi_{12}(\rho_{m})^{\otimes2}\Pi_{12}|\chi_{1}\rangle\otimes|\chi
_{i}\rangle}\right)  \right. \\
&  +\left.  |\langle\chi_{2}|\rho_{m}|\chi_{3}\rangle|-\sqrt{\langle\chi
_{2}|\otimes\langle\chi_{3}|\Pi_{23}(\rho_{m})^{\otimes2}\Pi_{23}|\chi
_{2}\rangle\otimes|\chi_{3}\rangle}\right] \\
&  =\sum\limits_{m}p_{m}\mathcal{T}(\rho_{m},\chi)\\
&  \leq0
\end{align*}
which implies that the inequality (\ref{eqth2}) holds for $k$-separable
$4$-partite states.

Notice that for any $k$-separable pure states $\rho=|\psi\rangle\langle\psi|$,
if the two different bits of $|\chi_{\alpha}\rangle$ and $|\chi_{\beta}%
\rangle$ are in two different parts, then $|\langle\chi_{\alpha}|\rho
|\chi_{\beta}\rangle|-\sqrt{\langle\chi_{\alpha}|\otimes\langle\chi_{\beta
}|\Pi_{\alpha\beta}\rho^{\otimes2}\Pi_{\alpha\beta}|\chi_{\alpha}%
\rangle\otimes|\chi_{\beta}\rangle}=0$, otherwise $|\langle\chi_{\alpha}%
|\rho|\chi_{\beta}\rangle|-\frac{\langle\chi_{\alpha}|\rho|\chi_{\alpha
}\rangle+\langle\chi_{\beta}|\rho|\chi_{\beta}\rangle}{2}\leq0$. This implies
that inequality (\ref{eqth2}) holds for $k$-separable pure $N$-partite states
$\rho$.

Suppose that $\rho=\sum\limits_{m}p_{m}\rho_{m}$ is a $k$-separable mixed
$N$-partite state, where $\rho_{m}=|\psi_{m}\rangle\langle\psi_{m}|$ is
$k$-separable. It follow that
\begin{align*}
\mathcal{T}(\rho,\chi)  &  =\sum\limits_{|\chi_{\alpha}\rangle\in V}%
\sum\limits_{|\chi_{\beta}\rangle\in K_{\alpha}}\left(  |\langle\chi_{\alpha
}|\sum\limits_{m}p_{m}\rho_{m}|\chi_{\beta}\rangle|-\sqrt{\langle\chi_{\alpha
}|\otimes\langle\chi_{\beta}|\Pi_{\alpha\beta}(\sum\limits_{m}p_{m}\rho
_{m})^{\otimes2}\Pi_{\alpha\beta}|\chi_{1}\rangle\otimes|\chi_{i}\rangle
}\right) \\
&  -N_{k}\sum\limits_{\alpha}\langle\chi_{\alpha}|\sum\limits_{m}p_{m}\rho
_{m}|\chi_{\alpha}\rangle\\
&  \leq\sum\limits_{m}p_{m}\left[  \sum\limits_{|\chi_{\alpha}\rangle\in
V}\sum\limits_{|\chi_{\beta}\rangle\in K_{\alpha}}\left(  |\langle\chi
_{\alpha}|\rho_{m}|\chi_{\beta}\rangle|-\sqrt{\langle\chi_{\alpha}%
|\otimes\langle\chi_{\beta}|\Pi_{\alpha\beta}(\rho_{m})^{\otimes2}\Pi
_{\alpha\beta}|\chi_{1}\rangle\otimes|\chi_{i}\rangle}\right)  \right. \\
&  \left.  -N_{k}\sum\limits_{\alpha}\langle\chi_{\alpha}|\rho_{m}%
|\chi_{\alpha}\rangle\right] \\
&  =\sum\limits_{m}p_{m}\mathcal{T}(\rho_{m},\chi),\\
&  \leq0
\end{align*}
which completes the proof.
\end{proof}

\section{Examples}

Consider the family of $N$-qubit mixed states%
\[%
\begin{tabular}
[c]{lll}%
$\rho^{(D_{2}^{N})}=a|D_{2}^{N}\rangle\langle D_{2}^{N}|+\frac{(1-a)I_{N}%
}{2^{N}},$ & where & $|D_{2}^{N}\rangle=\frac{1}{\sqrt{C_{N}^{2}}}%
\sum\limits_{1\leq i\neq j\leq N}|\phi_{i,j}\rangle.$%
\end{tabular}
\
\]
By Theorem 1, if%
\[
a>\frac{2C_{N}^{2}(N-2)+N_{k}C_{N}^{2}}{2C_{N}^{2}(N-2)+N_{k}C_{N}^{2}%
-2^{N}N_{k}+2^{N+1}(N-2)}%
\]
then $\rho^{(D_{2}^{N})}$ is $k$-nonseparable. Thus, if
\[
a>\frac{C_{N}^{2}(2N-5)}{C_{N}^{2}(2N-5)+2^{N}}%
\]
then $\rho^{(D_{2}^{N})}$ are genuine entangled, which is exactly the same as
in \cite{Huber11a}; if $a>\frac{9}{17}$ then $\rho^{(D_{2}^{4})}$ is genuine
entangled; if $a>\frac{5}{13}$ then $\rho^{(D_{2}^{4})}$ is 3-nonseparable; if
$a>\frac{3}{11}$ then $\rho^{(D_{2}^{4})}$ is not fully-separable; if
$a>\frac{5}{21}=0.23$ then $\rho^{(D_{2}^{5})}$ is not fully-separable.
However, the inequality in \cite{Huberqic} detects that, if $a>\frac{21}{29}$
then $\rho^{(D_{2}^{4})}$ is genuine entangled; if $a>\frac{9}{13}$ then
$\rho^{(D_{2}^{4})}$ is 3-nonseparable; if $a>\frac{3}{11}$ then $\rho
^{(D_{2}^{4})}$ is not fully-separable; if $a>0.27$ then $\rho^{(D_{2}^{5})}$
is not a fully-separable 5-partite state.

Consider the $N$-qubit state%
\[%
\begin{tabular}
[c]{lll}%
$\rho^{(D_{m}^{N})}=\frac{(1-a)I_{N}}{2^{N}}+a|D_{m}^{N}\rangle\langle
D_{m}^{N}|,$ & where & $|D_{m}^{N}\rangle=\frac{1}{\sqrt{C_{N}^{m}}}%
\sum\limits_{1\leq i_{1}\leq i_{2}\leq\cdots\leq i_{m}\leq N}|\phi
_{i_{1},i_{2},\cdots,i_{m}}\rangle,$%
\end{tabular}
\
\]
By Theorem 2, if
\[
a>\frac{mC_{N}^{m}(N-m)+N_{k}C_{N}^{m}}{mC_{N}^{m}(N-m)+N_{k}C_{N}^{m}%
-2^{N}N_{k}+2^{N}m(N-m)},
\]
then $\rho^{(D_{m}^{N})}$ is $k$-nonseparable. For $N=5$ and $m=3$, we get
that if $a>\frac{5}{13}$ then $\rho^{(D_{m}^{N})}$ is 3-nonseparable, while
the method in \cite{Huberqic} fails.

Consider the one-parameter four-qubit state%
\[%
\begin{tabular}
[c]{lll}%
$\rho=\frac{1-a}{16}I_{16}+a|\phi\rangle\langle\phi|,$ & where & $|\phi
\rangle=\frac{1}{2}(|0011\rangle+|0101\rangle+|0110\rangle+|1010\rangle).$%
\end{tabular}
\]
By Theorem 3, if $a>\frac{7}{19}$ and $a>\frac{1}{5}$ then $\rho$ is
3-nonseparable and not fully-separable, respectively, while in \cite{Huberqic}%
, if $a>\frac{9}{13}$ and $a>\frac{3}{11}$, then $\rho$ is 3-nonseparable and
not fully-separable.

\bigskip

\noindent\textbf{Acknowledgments}. This work is supported by the National
Natural Science Foundation of China under Grant 11371005,11371247, 11201427
and 11571313. SS would like to thank the Institute of Natural Sciences
(INS)\ at Shanghai Jiao Tong University for the kind hospitality during
completion of this work. The support of INS is gratefully acknowledged.



\end{document}